\documentclass[12pt]{article}
\pdfoutput=1

\usepackage{xcolor}
\definecolor{blue}{RGB}{0,70,140}
\definecolor{green}{RGB}{100,140,0}
\definecolor{red}{RGB}{190,10,50}

\colorlet{structurecolor}{red}
\colorlet{linkcolor}{blue}
\colorlet{citecolor}{green}

\usepackage[utf8]{inputenc}
\usepackage[T1]{fontenc}
\usepackage{mathpazo}

\usepackage[english]{babel}

\usepackage[a4paper,margin=2.5cm]{geometry}
\usepackage{parskip}

\newcommand{\affiliation}{\footnote}
\newcommand{\affiliationmark}[1][\value{footnote}-1]{\footnotemark[\numexpr#1+1\relax]}
\usepackage[shortlabels]{enumitem}
\setlist{itemsep=0ex,topsep=0ex,parsep=0.4ex}

\usepackage{amsmath,amsfonts,amssymb,amsthm,mathtools,thmtools,thm-restate,bbm,centernot}

\usepackage[hyphens]{url} % Line breaking of urls
\usepackage[hidelinks,colorlinks,pagebackref]{hyperref} % Coloring of hyperlinks and backreferences in bibliography; must be loaded after url
\hypersetup{citecolor=citecolor,linkcolor=linkcolor,urlcolor=linkcolor}
\usepackage[capitalise,nameinlink,noabbrev,compress]{cleveref} % References in the document; must be loaded after hyperref and amsmath

\renewcommand*{\backref}[1]{}
\renewcommand*{\backrefalt}[4]{
	\ifcase #1 Not cited.%
	\or $\uparrow$#2%
	\else $\uparrow$#2%
	\fi%
}

\declaretheorem[name=Definition,Refname={Definition,Definitions},numberwithin=section,style=definition]{definition}
\declaretheorem[name=Theorem,Refname={Theorem,Theorems},numberlike=definition,style=plain]{theorem}

\declaretheorem[name=Lemma,Refname={Lemma,Lemmas},numberlike=theorem,style=plain]{lemma}

\declaretheorem[name=Problem,Refname={Problem,Problems},numberlike=theorem,style=plain]{problem}

\newcommand{\defn}[1]{\textcolor{structurecolor}{\emph{#1}}}

\mathtoolsset{centercolon}

\let\emptyset\varnothing

\renewcommand{\l}{\mathopen{}\mathclose\bgroup\left}
\renewcommand{\r}{\aftergroup\egroup\right}

\newcommand{\st}{\ifnum\currentgrouptype=16 \mathrel{}\middle|\mathrel{}\else\mathrel{|}\fi}

\DeclarePairedDelimiter{\set}{\{}{\}}
\DeclarePairedDelimiter{\abs}{\lvert}{\rvert}

\DeclarePairedDelimiter{\floor}{\lfloor}{\rfloor}
\DeclarePairedDelimiter{\ceil}{\lceil}{\rceil}

\newcommand{\subs}{\subseteq}
\newcommand{\nsubs}{\nsubseteq}
\newcommand{\sups}{\supseteq}

\newcommand{\sm}{\setminus}
\newcommand{\eps}{\varepsilon}

\makeatletter
\newcommand{\generate}[4]{%
  \def\@tempa{#1} % we don't want to lowercase it
  \count@=`#3
  \loop
  \begingroup\lccode`?=\count@
  \lowercase{\endgroup\@namedef{\@tempa ?}{#2{?}}}%
  \ifnum\count@<`#4
  \advance\count@\@ne
  \repeat
}
\makeatother

\generate{b}{\mathbb}{A}{Z}
\generate{c}{\mathcal}{A}{Z}

\newcommand{\DeclareMath}[2]{\newcommand{#1}{\mathnormal{#2}}}

\newcommand{\DeclareBinaryMathOperator}[2]{\newcommand{#1}{\mathbin{#2}}}

\DeclareMath{\N}{\mathbb{N}}
\DeclareMath{\Z}{\mathbb{Z}}
\DeclareMath{\Q}{\mathbb{Q}}
\DeclareMath{\R}{\mathbb{R}}
\DeclareMath{\C}{\mathbb{C}}
\DeclareMath{\F}{\mathbb{F}}

\DeclareMath{\ind}{\mathbbm{1}}
\DeclareBinaryMathOperator{\divides}{|}

\DeclareMathOperator{\pr}{\bP}
\DeclareMathOperator{\ex}{\bE}

\title{Lower bounds for graph reconstruction with maximal independent set queries}
\author{Lukas Michel\affiliation{Mathematical Institute, University of Oxford, United Kingdom (\textsf{\{\href{mailto:michel@maths.ox.ac.uk}{michel},\href{mailto:scott@maths.ox.ac.uk}{scott}\}@maths.ox.ac.uk}). Research of Alex Scott supported by EPSRC grant EP/X013642/1.} \and Alex Scott\affiliationmark}
\date{4 April 2024}

\begin{document}
\maketitle

\begin{abstract}
    We investigate the number of maximal independent set queries required to reconstruct the edges of a hidden graph. We show that randomised adaptive algorithms need at least $\Omega(\Delta^2 \log(n / \Delta) / \log \Delta)$ queries to reconstruct $n$-vertex graphs of maximum degree $\Delta$ with success probability at least $1/2$, and we further improve this lower bound to $\Omega(\Delta^2 \log(n / \Delta))$ for randomised non-adaptive algorithms. We also prove that deterministic non-adaptive algorithms require at least $\Omega(\Delta^3 \log n / \log \Delta)$ queries.

    This improves bounds of Konrad, O'Sullivan, and Traistaru, and answers one of their questions. The proof of the lower bound for deterministic non-adaptive algorithms relies on a connection to cover-free families, for which we also improve known bounds.
\end{abstract}

\section{Introduction}

The graph reconstruction problem is the problem of determining the edges of a hidden graph $G = (V,E)$ through queries that reveal some partial information about the graph. For a given type of query, the following natural question presents itself.
\begin{quote}
    What is the minimum number of queries required to reconstruct any graph?
\end{quote}
This problem has been extensively studied for independent set queries that reveal whether a queried set of vertices contains an edge. Bounds on the minimum number of queries required were studied for general graphs \cite{AC08,AB19}, specific families of graphs \cite{AA05,ABKRS04}, and hypergraphs \cite{AC06,ABM14}. In related graph problems, other query models were considered. For instance, maximal matching queries were used to approximate maximum matchings in a graph \cite{KK20,KNS23}, and vertex degree queries were used to estimate the average degree of a graph \cite{F06,GR08}.

Konrad, O'Sullivan, and Traistaru \cite{KOT24} recently investigated graph reconstruction for maximal independent set queries. In this setting, an algorithm can query a set of vertices $Q \subs V$ to obtain a maximal independent set of the subgraph of $G$ induced by $Q$. Using this information, the algorithm has to reconstruct all edges of the graph.

There are two important properties that influence the minimum number of queries required by such an algorithm. Firstly, queries can either be adaptive, which means that they depend on the outcomes of other queries, or they are non-adaptive if they do not. While adaptive algorithms require fewer queries, non-adaptive algorithms are desirable because their queries can be performed in parallel. Secondly, algorithms can randomise their queries. In contrast to deterministic algorithms which always reconstruct the graph correctly, randomised algorithms only succeed with high probability, but they might use far fewer queries.

\textbf{Randomised algorithms.} One way of choosing non-adaptive queries for a randomised algorithm is to choose each query by independently including each vertex with a fixed probability. By analysing this strategy, Konrad, O'Sullivan, and Traistaru \cite{KOT24} showed that $\cO(\Delta^2 \log n)$ non-adaptive queries suffice to reconstruct graphs with $n$ vertices and maximum degree $\Delta$ with high probability.

\begin{theorem}[Konrad, O'Sullivan, and Traistaru]\label{thm:randnonadupperbound}
    There is a randomised non-adaptive algorithm that uses $\cO(\Delta^2 \log n)$ queries to reconstruct any graph of maximum degree $\Delta$ with high probability.
\end{theorem}

To complement this upper bound, Konrad, O'Sullivan, and Traistaru proved that any randomised adaptive algorithm that succeeds with probability at least $1/2$ requires at least $\Omega(\Delta^2 + \log n)$ queries. They asked whether an algorithm could attain this lower bound. We show that this is not possible.

\begin{theorem}\label{thm:randadlowerbound}
    The number of queries of a randomised adaptive algorithm that reconstructs any graph of maximum degree $\Delta$ with probability at least $1/2$ is at least
    \[
        \Omega\l(\frac{\Delta^2 \log\l(\frac{n}{\Delta}\r)}{\log \Delta}\r).
    \]
\end{theorem}

For non-adaptive algorithms, we further improve this lower bound by a factor of $\log \Delta$. This matches \cref{thm:randnonadupperbound} for $\Delta \le n^{1-\eps}$ when $\eps > 0$ is fixed.

\begin{theorem}\label{thm:randnonadlowerbound}
    The number of queries of a randomised non-adaptive algorithm that reconstructs any graph of maximum degree $\Delta$ with probability at least $1/2$ is at least
    \[
        \Omega\l(\Delta^2 \log\l(\frac{n}{\Delta}\r)\r).
    \]
\end{theorem}

The proofs of these lower bounds rely on graphs that contain a clique of size $\Theta(\Delta)$ while the remaining vertices form an independent set. We will show that maximal independent set queries reveal only very little information about such graphs. Counting the number of these graphs then gives the desired lower bounds.

\textbf{Deterministic algorithms.} To obtain a deterministic algorithm, Konrad, O'Sullivan, and Traistaru \cite{KOT24} adapted the idea of their randomised algorithm. They showed that a fixed collection of sufficiently many random queries will reconstruct every graph correctly. This gives a deterministic algorithm with $\cO(\Delta^3 \log n)$ non-adaptive queries.\footnote{In fact, their arguments show that $\cO(\Delta^3 \log(n / \Delta))$ queries are sufficient.}

\begin{theorem}[Konrad, O'Sullivan, and Traistaru]\label{thm:detnonadupperbound}
    There is a deterministic non-adaptive algorithm that uses $\cO(\Delta^3 \log n)$ queries to reconstruct any graph of maximum degree $\Delta$.
\end{theorem}

Konrad, O'Sullivan, and Traistaru also showed that every deterministic non-adaptive algorithm needs at least $\Omega(\Delta^3 / (\log \Delta)^2 + \log n)$ queries. We provide a lower bound that matches the upper bound from \cref{thm:detnonadupperbound} up to a factor of $\log \Delta$.

\begin{theorem}\label{thm:detnonadlowerbound}
    The number of queries of a deterministic non-adaptive algorithm that reconstructs any graph of maximum degree $\Delta$ is at least
    \[
        \Omega\l(\min\set*{n^2, \frac{\Delta^3 \log n}{\log \Delta}}\r).
    \]
\end{theorem}

Our proof of \cref{thm:detnonadlowerbound} relates deterministic non-adaptive algorithms to cover-free families. Here, a family of sets $\cF$ is \defn{$(w,r)$-cover-free} if for all $A_1, \dots, A_w \in \cF$ and all $B_1, \dots, B_r \in \cF \sm \set{A_1, \dots, A_w}$ we have
\[
    \bigcap_{i=1}^w A_i \nsubs \bigcup_{i=1}^r B_i.
\]
The size of cover-free families has been widely studied because of their relevance to non-adaptive group testing \cite{HS87,DH00,DH06,AJS19}, key distribution patterns \cite{MP88,DFFT95}, superimposed codes \cite{KS64,DR82}, and more \cite{W06,E96}. For a given $n$, let $t(n,w,r)$ denote the minimal $t$ such that there exists a $(w,r)$-cover-free family $\cF \subs \cP(t)$ with $n$ sets. We observe the following connection between cover-free families and deterministic non-adaptive algorithms.

\begin{lemma}\label{lem:detnonadboundscoverfree}
    The minimum number of queries of a deterministic non-adaptive algorithm that reconstructs any graph of maximum degree $\Delta$ is at least $t(n,2,2\Delta-2)$ and at most $t(n,2,2\Delta)$.
\end{lemma}

Upper and lower bounds on $t(n,w,r)$ have been extensively studied. On the side of upper bounds, it is well-known \cite{DFFT95,E96,STW00,SW04} that
\[
    t(n,w,r) = \cO\l(\frac{(w+r)^{w+r+1}}{w^w r^r} \log n\r).
\]
If $w$ is fixed, this yields $t(n,w,r) = \cO(r^{w+1} \log n)$ which gives an alternative proof of \cref{thm:detnonadupperbound}.

In terms of lower bounds, when $w = 1$, D'yachkov and Rykov \cite{DR82} proved that $t(n,1,r) = \Omega(r^2 \log n / \log r)$ if $r$ is fixed and $n \to \infty$. Later, Ruszink\'{o} \cite{R94} and F\"{u}redi \cite{F96} gave purely combinatorial proofs of this result. More generally, for arbitrary $r$ and $n$, the argument of F\"{u}redi gives the following lower bound.

\begin{theorem}[F\"{u}redi]\label{thm:tn1r}
    If $1 \le r < n$, then
    \[
        t(n,1,r) = \Omega\l(\min\set*{n, \frac{r^2 \log n}{\log r}}\r).
    \]
\end{theorem}

This shows that $t(n,1,r) = \Omega(r^2 \log n / \log r)$ for $r = \cO(\sqrt{n})$, matching the upper bound up to a factor of $\log r$, and $t(n,1,r) = \Theta(n)$ for $r = \Omega(\sqrt{n})$.

For $w \ge 2$, previous lower bounds were more complicated. Stinson, Wei, and Zhu \cite{SWZ00} showed that $t(n,w,r) = \Omega(r^{w+1} \log n / \log r)$ if $r$ is fixed and $n \to \infty$. It was only later determined by Ma and Wei \cite{MW04} that this result holds for $r = \cO(\sqrt{n})$. For $r = \Omega(\sqrt{n})$, Abdi and Bshouty \cite{AB16} proved the best lower bounds known. They showed that $t(n,w,r) = \Omega(r^{w+1} \log n / \log^{k+1} r)$ for $\Omega(n^{(k-1)/k}) \le r \le \cO(n^{k/(k+1)})$ where $2 \le k \le w$, and $t(n,w,r) = \Theta(n^w)$ for $r = \Omega((n \log n)^{w/(w+1)})$.

Using probabilistic arguments, we generalise \cref{thm:tn1r} to $w \ge 2$. This improves the previous lower bounds for $r = \Omega(\sqrt{n})$.

\begin{theorem}\label{thm:tnwr}
    If $w \ge 1$ is fixed and $1 \le r < n$, then
    \[
        t(n,w,r) = \Omega\l(\min\set*{n^w, \frac{r^{w+1} \log n}{\log r}}\r).
    \]
\end{theorem}

This result implies that $t(n,w,r) = \Omega(r^{w+1} \log n / \log r)$ for $r = \cO(n^{w/(w+1)})$ and $t(n,w,r) = \Theta(n^w)$ for $r = \Omega(n^{w/(w+1)})$. \cref{thm:detnonadlowerbound} is an immediate consequence of this result in combination with \cref{lem:detnonadboundscoverfree}.

The rest of the paper is organised as follows. In \cref{sec:coverfree} we prove \cref{thm:tnwr}, our lower bound for cover-free families. We use this result in \cref{sec:detalg} to obtain a lower bound on the number of queries of deterministic algorithms, proving \cref{thm:detnonadlowerbound,lem:detnonadboundscoverfree}. In \cref{sec:randalg}, we then consider randomised algorithms and prove \cref{thm:randadlowerbound,thm:randnonadlowerbound}. We finish with some open problems in \cref{sec:openproblems}.

\textbf{Notation.} We denote the dual of a family of sets $\cF \subs \cP(X)$ by $\cF^* = \set{\cF_x : x \in X}$ where $\cF_x = \set{F \in \cF : x \in F}$. Moreover, if $G = (V, E)$ is a graph and $X \subs V$, then $G[X]$ is the subgraph of $G$ induced by $X$, and $N_G(X)$ is the neighbourhood of $X$ in $G$.

\section{Cover-free families}
\label{sec:coverfree}

In this section we prove \cref{thm:tnwr} and show that for any fixed $w \ge 1$,
\[
    t(n,w,r) = \Omega\l(\min\set*{n^w, \frac{r^{w+1} \log n}{\log r}}\r).
\]
We split the proof of this result into two regimes. If $r = \cO(\sqrt{n})$, we use \cref{thm:tn1r} to show that $t(n,w,r) = \Omega(r^{w+1} \log n / \log r)$. We will do this with the following result.

\begin{lemma}\label{lem:reducetnwrtotn1r}
    If $w, r, s \ge 1$ and $w + r \le n$, it holds that
    \[
        t(n,w+1,r+s) \ge \l(\frac{w + r}{w}\r)^w \cdot t(n-w-r,1,s).
    \]
\end{lemma}

\begin{proof}
    Let $\cF \subs \cP(t)$ be a $(w+1,r+s)$-cover-free family with $n$ sets, and suppose that
    \[
        t < \l(\frac{w + r}{w}\r)^w \cdot t(n-w-r,1,s).
    \]
    Pick $A_1, \dots, A_w \in \cF$ uniformly and independently at random. Let $\cA = \set{A_1, \dots, A_w}$. Then, pick $B_1, \dots, B_r \in \cF \sm \cA$ uniformly and independently at random and let $\cB = \set{B_1, \dots, B_r}$. For any $x \in [t]$, if $\alpha_x = \abs{\cF_x} / n$, we have
    \[
        \pr(\cA \subs \cF_x, \cB \subs \cF_x^c) \le \l(\frac{\abs{\cF_x}}{n}\r)^w \l(\frac{n - \abs{\cF_x}}{n - w}\r)^r = \l(\frac{n}{n-w}\r)^r \alpha_x^w \l(1 - \alpha_x\r)^r.
    \]
    Note that the term $\alpha_x^w \l(1 - \alpha_x\r)^r$ is maximised by $\alpha_x = w / (w + r)$, and so we get
    \[
        \pr(\cA \subs \cF_x, \cB \subs \cF_x^c) \le \l(\frac{n}{n - w}\r)^r \l(\frac{w}{w + r}\r)^w \l(\frac{r}{w + r}\r)^r \le \l(\frac{w}{w + r}\r)^w,
    \]
    where the last inequality used the fact that $w + r \le n$. Thus, if
    \[
        X = \set{x \in [t] : \cA \subs \cF_x, \cB \subs \cF_x^c},
    \]
    then $\ex(\abs{X}) \le (w / (w + r))^w \cdot t < t(n-w-r,1,s)$, and so with positive probability we have $\abs{X} < t(n-w-r,1,s)$. This implies that the sets $F \cap X$ with $F \in \cF \sm (\cA \cup \cB)$ cannot form a $(1,s)$-cover-free family\footnote{Note that it could happen that $F \cap X = F' \cap X$ for some distinct $F, F' \in \cF \sm (\cA \cup \cB)$. In that case, we could simply pick $A_{w+1} = F$ and $C_1 = \dots = C_s = F'$ in the argument that follows.}, and so there exist $A_{w+1} \in \cF \sm (\cA \cup \cB)$ and $C_1, \dots, C_s \in \cF \sm (\cA \cup \cB \cup \set{A_{w+1}})$ such that
    \[
        A_{w+1} \cap X \subs \bigcup_{i=1}^s (C_i \cap X) \subs \bigcup_{i=1}^s C_i.
    \]
    Now, suppose that $x \in \bigcap_{i=1}^{w+1} A_i$. If $x \notin X$, we must have $x \in \bigcup_{i=1}^r B_i$. Otherwise $x \in X$, so $x \in A_{w+1} \cap X$ and therefore $x \in \bigcup_{i=1}^s C_i$. So,
    \[
        \bigcap_{i=1}^{w+1} A_i \subs \bigcup_{i=1}^r B_i \cup \bigcup_{i=1}^s C_i.
    \]
    This contradicts the fact that $\cF$ is $(w+1,r+s)$-cover-free.
\end{proof}

If $r = \Omega(\sqrt{n})$, then $\log r = \Theta(\log n)$, and so the lower bound we want to show simplifies to $t(n,w,r) = \Omega(\min\set{n^w, r^{w+1}})$. Using arguments very similar to the proof of the previous lemma, we show that this is true.

\begin{lemma}\label{lem:tnwrlowerbound}
    If $w, r, s \ge 1$ and $w + r \le n$, it holds that
    \[
        t(n,w,r+s) \ge \min\set*{\frac{1}{2} \l(\frac{n}{w}\r)^w, \frac{s}{2} \l(\frac{w + r}{w}\r)^w}.
    \]
\end{lemma}

\begin{proof}
    Let $\cF \subs \cP(t)$ be a $(w,r+s)$-cover-free family with $n$ sets, and suppose that
    \[
        t < \min\set*{\frac{1}{2} \l(\frac{n}{w}\r)^w, \frac{s}{2} \l(\frac{w + r}{w}\r)^w}.
    \]
    We proceed as in the proof of \cref{lem:reducetnwrtotn1r}. Pick $A_1, \dots, A_w \in \cF$ uniformly at random and let $\cA = \set{A_1, \dots, A_w}$. Then, pick $B_1, \dots, B_r \in \cF \sm \cA$ uniformly at random and let $\cB = \set{B_1, \dots, B_r}$. For any $x \in [t]$, we have
    \[
        \pr(\cA \subs \cF_x, \cB \subs \cF_x^c) \le \l(\frac{w}{w + r}\r)^w.
    \]
    Thus, if $X = \set{x \in [t] : \cA \subs \cF_x, \cB \subs \cF_x^c}$, then $\ex(\abs{X}) \le (w / (w + r))^w \cdot t \le s / 2$, and so $\pr(\abs{X} > s) \le 1/2$. Moreover, the event $\set{\cA = \cF_x}$ can only happen if $\abs{\cF_x} \le w$, and in that case we have $\pr(\cA = \cF_x) \le (w / n)^w$. So,
    \[
        \pr(\cA = \cF_x \text{ for some } x \in [t]) \le t \cdot \l(\frac{w}{n}\r)^w < \frac{1}{2}.
    \]
    This implies that with positive probability we have $\abs{X} \le s$ and $\cA \neq \cF_x$ for all $x \in [t]$. Note that for each $x \in X$ we have $\cA \subs \cF_x$, and so there must exist some $C_x \in \cF_x \sm \cA$.
    
    Now, suppose that $x \in \bigcap_{i=1}^w A_i$. If $x \notin X$, we must have $x \in \bigcup_{i=1}^r B_i$. Otherwise $x \in X$, and so $x \in C_x$. Therefore,
    \[
        \bigcap_{i=1}^w A_i \subs \bigcup_{i=1}^r B_i \cup \bigcup_{x \in X} C_x.
    \]
    This contradicts the fact that $\cF$ is $(w,r+s)$-cover-free.
\end{proof}

Putting these two lemmas together, we obtain the desired result.

\begin{proof}[Proof of \cref{thm:tnwr}]
    If $r = \cO(\sqrt{n})$, \cref{lem:reducetnwrtotn1r} together with \cref{thm:tn1r} implies that
    \[
        t(n,w,r) \ge \l(\frac{w-1+\ceil*{\frac{r}{2}}}{w-1}\r)^{w-1} \cdot t\l(n+1-w-\ceil*{\frac{r}{2}},1,\floor*{\frac{r}{2}}\r) = \Omega\l(\frac{r^{w+1} \log n}{\log r}\r).
    \]
    Otherwise, if $r = \Omega(\sqrt{n})$, \cref{lem:tnwrlowerbound} shows that
    \[
        t(n,w,r) \ge \min\set*{\frac{1}{2} \l(\frac{n}{w}\r)^w, \frac{\floor*{\frac{r}{2}}}{2} \l(\frac{w + \ceil*{\frac{r}{2}}}{w}\r)^w} = \Omega\l(\min\set*{n^w, \frac{r^{w+1} \log n}{\log r}}\r). \qedhere
    \]
\end{proof}

\section{Deterministic non-adaptive algorithms}
\label{sec:detalg}

In this section, we use the lower bound for cover-free families to prove a lower bound on the number of queries of deterministic non-adaptive algorithms. Formally, such an algorithm corresponds to a family of queries $\cQ \subs \cP(V)$ with the following property: for all distinct graphs $G$ and $H$ with maximum degree $\Delta$ there exists a query $Q \in \cQ$ such that no set of vertices is a maximal independent set of both $G[Q]$ and $H[Q]$. We call such a family $\cQ$ a \defn{query scheme}.

To bound the number of queries in such schemes, we relate them to cover-free families. Our main observation is that a family of queries $\cQ \subs \cP(V)$ is a query scheme if and only if its dual $\cQ^*$ is $(2,r)$-cover-free for $r \approx 2 \Delta$. This was essentially already proved by Konrad, O'Sullivan, and Traistaru \cite{KOT24}, but with different terminology. We provide a proof of the direction necessary for the lower bound in \cref{lem:detnonadboundscoverfree}.

\begin{lemma}\label{lem:reconstratcoverfree}
    If $\cQ \subs \cP(V)$ is a query scheme, then $\cQ^*$ is $(2,2\Delta-2)$-cover-free.
\end{lemma}

\begin{proof}
    If $\cQ^*$ is not $(2,2\Delta-2)$-cover-free, there exist two distinct vertices $u, v \in V$ and a set of vertices $W \subs V \sm \set{u, v}$ with $\abs{W} \le 2 \Delta - 2$ such that $\cQ_u \cap \cQ_v \subs \bigcup_{w \in W} \cQ_w$. Let $G$ be a graph with the least number of edges that has maximum degree $\Delta - 1$ and satisfies $N_G(\set{u, v}) = W$. Let $H = G \cup \set{uv}$.
    
    For every $Q \in \cQ$, let $I_Q$ be a maximal independent set of $G[Q]$ with $Q \sm \set{u, v} \subs I_Q$. We claim that $I_Q$ is also a maximal independent set of $H[Q]$. Indeed, otherwise we must have $u, v \in I_Q \subs Q$. Then, $Q \in \cQ_u \cap \cQ_v \subs \bigcup_{w \in W} \cQ_w$, and so $W \cap Q \neq \emptyset$. This implies that
    \[
        N_G(\set{u, v}) \cap I_Q = W \cap I_Q \sups W \cap (Q \sm \set{u, v}) = W \cap Q \neq \emptyset.
    \]
    Since $u, v \in I_Q$, this contradicts the fact that $I_Q$ is an independent set of $G$. So, for all $Q \in \cQ$, $I_Q$ is a maximal independent set of both $G[Q]$ and $H[Q]$. This contradicts the fact that $\cQ$ is a query scheme.
\end{proof}

Using arguments from Konrad, O'Sullivan, and Traistaru \cite{KOT24}, we can obtain the following converse.

\begin{lemma}\label{lem:coverfreereconstrat}
    Let $\cQ \subs \cP(V)$. If $\cQ^*$ is $(2,2\Delta)$-cover-free, then $\cQ$ is a query scheme.
\end{lemma}

Using these two results, we can now prove \cref{lem:detnonadboundscoverfree}. Together with \cref{thm:tnwr}, this proves \cref{thm:detnonadlowerbound}.

\begin{proof}[Proof of \cref{lem:detnonadboundscoverfree}]
    Let $\cQ \subs \cP(V)$ be a query scheme. Then, by \cref{lem:reconstratcoverfree}, we know that $\cQ^* \subs \cP(\cQ)$ is $(2,2\Delta-2)$-cover-free, and so
    \[
        \abs{\cQ} \ge t(\abs{\cQ^*},2,2\Delta-2) = t(n,2,2\Delta-2).
    \]
    On the other hand, if $t = t(n,2,2\Delta)$, there exist a family $\cR \subs \cP(t)$ with $\abs{\cR} = n$ that is $(2,2\Delta)$-cover-free. By identifying $\cR$ with $V$, we obtain a family $\cQ = \cR^* \subs \cP(V)$ with $\abs{\cQ} = t$ such that $\cQ^* \cong \cR$ is $(2,2\Delta)$-cover-free. Then, by \cref{lem:coverfreereconstrat}, $\cQ$ is a query scheme with $t = t(n,2,2\Delta)$ queries.
\end{proof}

\section{Randomised algorithms}
\label{sec:randalg}

Finally, we want to lower bound the number of queries of randomised algorithms, both in the adaptive and non-adaptive setting. We begin with some general observations.

First, note that any randomised algorithm $\cA$ can randomly decide in advance for every possible execution path which set of vertices it wants to query next. This gives a deterministic strategy $S$ that the algorithm will then perform to reconstruct the graph. From now on, we will assume that every algorithm is of this form. That is, $\cA$ randomly selects in advance a deterministic strategy $S \in \cS$ from some family of strategies $\cS$, and it will then execute this strategy to reconstruct the graph.

The following lemma tells us that by analysing the strategies $S \in \cS$ separately, we can show that there is a graph that will likely be incorrectly reconstructed by $\cA$.

\begin{lemma}\label{lem:strattoalg}
    Let $\cG$ be a family of graphs and $\cS$ be a family of strategies of a randomised algorithm $\cA$. If $G \in \cG$ is a random graph, then there exists a graph $H \in \cG$ such that
    \[
        \pr(\cA \text{ reconstructs } H) \le \sup_{S \in \cS} \pr(S \text{ reconstructs } G).
    \]
\end{lemma}

\begin{proof}
    Let $p = \sup_{S \in \cS} \pr(S \text{ reconstructs } G)$. Denote the indicator variable of an event $E$ by $\ind_E$. Then, it holds that
    \begin{align*}
        \ex\l(\ind_{\set{\cA \text{ reconstructs } G}}\r) & = \sum_{S \in \cS} \pr(\cA \text{ selects } S) \cdot \ex\l(\ind_{\set{S \text{ reconstructs } G}}\r) \\
        & = \sum_{S \in \cS} \pr(\cA \text{ selects } S) \cdot \pr(S \text{ reconstructs } G) \\
        & \le \sum_{S \in \cS} \pr(\cA \text{ selects } S) \cdot p = p.
    \end{align*}
    In particular, there exists a graph $H \in \cG$ such that
    \[
        \pr(\cA \text{ reconstructs } H) = \ex\l(\ind_{\set{\cA \text{ reconstructs } H}}\r) \le p. \qedhere
    \]
\end{proof}

\subsection{Randomised adaptive algorithms}

We start by providing a lower bound on the number of queries of a randomised adaptive algorithm. For this, we will consider the family of graphs that contain a clique $U$ of size $\Theta(\Delta)$ while $V \sm U$ is an independent set.

Then, for every query $Q \subs V$ of the algorithm, either $Q \cap (V \sm U)$ is a maximal independent set, or there is a vertex $u \in Q \cap U$ such that $(Q \cap (V \sm U)) \cup \set{u}$ is a maximal independent set. So, the answer to a query will always be one out of at most $\Theta(\Delta)$ many alternatives. This implies that with $t$ queries, the algorithm can correctly reconstruct at most $\cO(\Delta^t)$ many graphs. By counting the number of graphs in our family, this yields the desired lower bound on $t$.

\begin{proof}[Proof of \cref{thm:randadlowerbound}]
    Let $U \subs V$ be a set with $\abs{U} = \ceil{\Delta / 2}$ and let $\cG$ be the family of graphs $G$ with maximum degree $\Delta$ such that $U$ is a clique of $G$ and $V \sm U$ is an independent set of $G$. Since every vertex of $U$ can have up to $\Delta - (\abs{U} - 1)$ neighbours in $V \sm U$, we get
    \[
        \abs{\cG} \ge \binom{n - \abs{U}}{\Delta - (\abs{U} - 1)}^{\abs{U}} \ge \binom{n - \Delta}{\ceil{\frac{\Delta}{2}}}^{\ceil{\frac{\Delta}{2}}} \ge \l(\frac{n - \Delta}{\ceil{\frac{\Delta}{2}}}\r)^{\ceil{\frac{\Delta}{2}}^2} \ge \l(\frac{n - \Delta}{\Delta}\r)^{\frac{\Delta^2}{4}}.
    \]
    Now, let $\cS$ be the family of strategies of a randomised adaptive algorithm $\cA$ with success probability at least $1/2$, and suppose that $\cA$ makes at most $t$ queries. Pick a graph $G \in \cG$ uniformly at random, and fix a strategy $S \in \cS$.
    
    For every possible query $Q \subs V$, note that either $Q \cap (V \sm U)$ is a maximal independent set of $G[Q]$, or there is a vertex $u \in Q \cap U$ such that $(Q \cap (V \sm U)) \cup \set{u}$ is a maximal independent set of $G[Q]$. So, when executing $S$ on graphs from $\cG$, the answer to each query of $S$ is always one out of $\abs{Q \cap U} + 1 \le \Delta + 1$ different sets. Since $S$ is deterministic and makes at most $t$ queries, it follows that $S$ will correctly reconstruct at most $(\Delta + 1)^t$ graphs from $\cG$. Therefore,
    \[
        \pr(S \text{ reconstructs } G) \le \frac{(\Delta + 1)^t}{\abs{\cG}}.
    \]
    Given that the success probability of $\cA$ is at least $1/2$, it follows by \cref{lem:strattoalg} that $(\Delta + 1)^t / \abs{\cG} \ge 1/2$. If $\Delta \le n/3$, this implies that
    \[
        t \ge \frac{\log\l(\frac{\abs{\cG}}{2}\r)}{\log(\Delta + 1)} \ge \frac{\Delta^2 \log\l(\frac{n - \Delta}{\Delta}\r) - 4 \log 2}{4 \log(\Delta + 1)} = \Omega\l(\frac{\Delta^2 \log\l(\frac{n}{\Delta}\r)}{\log \Delta}\r).
    \]
    For $\Delta \ge n/3$, we can simply apply this argument with $\Delta = n/3$ to obtain the desired lower bound.
\end{proof}

\subsection{Randomised non-adaptive algorithms}

To obtain a lower bound on the number of queries of a randomised non-adaptive algorithm, we adapt the arguments of the proof from the last section. In this proof, recall that the answer to a query $Q$ of the algorithm is always one out of $\abs{Q \cap U} + 1$ many alternatives. We used that this is at most $\Delta + 1$, but if $Q \cap U$ is much smaller than this for most queries, then we should get a better lower bound on the number of queries.

For adaptive algorithms, we cannot hope for this to be true because such algorithms can easily identify $U$ with the first query and then always ensure that $U \subs Q$. For non-adaptive algorithms, however, a randomly chosen $U$ will usually have a very small intersection with a query $Q$ of the algorithm, unless $Q$ is large. But if $Q$ is large, then most vertices of $U$ will have a neighbour in $Q$ which again restricts the number of possible answers to the query. Using this idea, we can improve the lower bound on the number of queries for non-adaptive algorithms by a factor of $\log \Delta$.

\begin{proof}[Proof of \cref{thm:randnonadlowerbound}]
    For disjoint sets $U, W \subs V$ with $\abs{U} = \ceil{\Delta / 3}$ and $\abs{W} = \floor{\Delta / 3}$, let $\cG_{U,W}$ be the family of graphs $G$ with maximum degree $\Delta$ such that $U$ is a clique of $G$, $V \sm U$ is an independent set of $G$, and all vertices in $U$ are connected to all vertices in $W$. Let $\cG$ be the union of all the families $\cG_{U,W}$. Note that
    \[
        N = \abs{\cG_{U,W}} \ge \binom{n - \abs{U} - \abs{W}}{\Delta - (\abs{U} - 1) - \abs{W}}^{\abs{U}} \ge \binom{n - \Delta}{\ceil{\frac{\Delta}{3}}}^{\ceil{\frac{\Delta}{3}}} \ge \l(\frac{n - \Delta}{\Delta}\r)^{\frac{\Delta^2}{9}}.
    \]
    Let $\cS$ be the family of strategies of a randomised non-adaptive algorithm $\cA$ with success probability at least $1/2$, and suppose that $\cA$ makes at most $t$ queries. Pick a graph $G \in \cG$ at random by first picking the sets $U,W \subs V$ uniformly at random and then picking $G \in \cG_{U,W}$ uniformly at random. Fix a strategy $S \in \cS$. Since $\cA$ is non-adaptive, the strategy $S$ corresponds of a family of queries $\cQ \subs \cP(V)$ with $\abs{\cQ} \le t$.

    For every query $Q \in \cQ$, note that $Q \cap (V \sm U)$ is a maximal independent set of $G[Q]$ whenever $Q \cap W \neq \emptyset$. Otherwise, either $Q \cap (V \sm U)$ or $(Q \cap (V \sm U)) \cup \set{u}$ for some $u \in Q \cap U$ is a maximal independent set of $G[Q]$. So, when executing $S$ on graphs from $\cG_{U,W}$, the answer to each query $Q \in \cQ$ is always one out of $D_Q$ different sets, where $D_Q = 1$ if $Q \cap W \neq \emptyset$ and $D_Q = \abs{Q \cap U} + 1$ otherwise. It follows that, $S$ will correctly reconstruct at most $D = \prod_{Q \in \cQ} D_Q$ graphs from $\cG_{U,W}$. Therefore, for any $M \ge 0$,
    \[
        \pr\l(S \text{ reconstructs } G \st D \le M\r) \le \frac{M}{N}.
    \]
    Note that
    \[
        D_Q = 1 + \sum_{u \in U} \ind_{\set{u \in Q, W \subs Q^c}}.
    \]
    Therefore,
    \[
        \ex(\log D_Q) \le \ex(D_Q) \le 1 + \sum_{u \in U} \ex\l(\ind_{\set{u \in Q, W \subs Q^c}}\r) = 1 + \sum_{u \in U} \pr(u \in Q, W \subs Q^c).
    \]
    Let $\alpha_Q = \abs{Q} / n$. Since $U$ and $W$ are chosen uniformly at random, we have
    \[
        \pr(u \in Q, W \subs Q^c) \le \frac{\abs{Q}}{n} \l(\frac{n-\abs{Q}}{n-1}\r)^{\abs{W}} = \l(\frac{n}{n-1}\r)^{\abs{W}} \alpha_Q (1 - \alpha_Q)^{\abs{W}}.
    \]
    The term $\alpha_Q (1 - \alpha_Q)^{\abs{W}}$ is maximised by $\alpha_Q = 1 / (\abs{W} + 1)$, and so we get
    \[
        \pr(u \in Q, W \subs Q^c) \le \l(\frac{n}{n-1}\r)^{\abs{W}} \frac{1}{\abs{W}+1} \l(\frac{\abs{W}}{\abs{W}+1}\r)^{\abs{W}} \le \frac{1}{\abs{W}+1} \le \frac{3}{\Delta}.
    \]
    It follows that $\ex(\log D_Q) \le 1 + 3 \abs{U} / \Delta \le 4$ and therefore, by Markov's inequality,
    \[
        \pr\l(D \ge e^{16 t}\r) = \pr\l(\sum_{Q \in \cQ} \log D_Q \ge 16 t\r) \le \frac{\ex(\sum_{Q \in \cQ} \log D_Q)}{16 t} \le \frac{4 t}{16 t} = \frac{1}{4}.
    \]
    Overall, we get
    \begin{align*}
        \pr(S \text{ reconstructs } G) & \le \pr\l(D \ge e^{16 t}\r) + \pr\l(S \text{ reconstructs } G \st D \le e^{16 t}\r) \pr\l(D \le e^{16 t}\r) \\
        & \le \frac{1}{4} + \frac{e^{16 t}}{N}.
    \end{align*}
    Given that the success probability of $\cA$ is at least $1/2$, it follows by \cref{lem:strattoalg} that $e^{16 t} / N \ge 1/4$. If $\Delta \le n/3$, this implies that
    \[
        t \ge \frac{1}{16} \log\l(\frac{N}{4}\r) \ge \frac{\Delta^2}{144} \log\l(\frac{n - \Delta}{\Delta}\r) - \frac{1}{16} \log 4 = \Omega\l(\Delta^2 \log\l(\frac{n}{\Delta}\r)\r).
    \]
    For $\Delta \ge n/3$, we can simply apply this argument with $\Delta = n/3$ to obtain the desired lower bound.
\end{proof}

\section{Open problems}
\label{sec:openproblems}

In this paper, we improved the lower bounds on the number of queries required for graph reconstruction with maximal independent set queries. However, some gaps between the lower and upper bounds remain.

For deterministic non-adaptive algorithms, we showed that $\Omega(\Delta^3 \log n / \log \Delta)$ queries are necessary while the upper bound $\cO(\Delta^3 \log n)$ was given by Konrad, O'Sullivan, and Traistaru \cite{KOT24}. This is the same gap as the longstanding gap between the lower and upper bounds $\Omega(r^{w+1} \log n / \log r) \le t(n,w,r) \le \cO(r^{w+1} \log n)$ for cover-free families.

\begin{problem}
    Are there $(1,r)$-cover-free families $\cF \subs \cP(t)$ with $t \in \cO(r^2 \log \abs{\cF} / \log r)$?
\end{problem}

For randomised non-adaptive algorithms, the upper bound of Konrad, O'Sullivan, and Traistaru \cite{KOT24} is $\cO(\Delta^2 \log n)$. We proved that such algorithms need at least $\Omega(\Delta^2 \log(n / \Delta))$ queries, and ask whether this is tight.

\begin{problem}
    Is there a randomised non-adaptive algorithm which uses $\cO(\Delta^2 \log(n / \Delta))$ queries to reconstruct any graph of maximum degree $\Delta$ with high probability?
\end{problem}

Finally, for adaptive algorithms, we have the same upper bounds as for non-adaptive algorithms in both the deterministic and randomised setting. However, our best lower bound drops to $\Omega(\Delta^2 \log(n / \Delta) / \log \Delta)$ in both settings. This suggests that adaptive algorithms might perform better than their non-adaptive counterparts.

\begin{problem}
    Is there a deterministic adaptive algorithm that uses $o(\Delta^3 \log(n/\Delta))$ queries to reconstruct any graph of maximum degree $\Delta$?
\end{problem}

{
    \fontsize{11pt}{12pt}
    \selectfont
    
    \hypersetup{linkcolor=structurecolor}
\newcommand{\etalchar}[1]{$^{#1}$}

}
\end{document}